\documentclass[a4paper,11pt]{article}
\usepackage[left=1.6cm,top=1.6cm, bottom=1.6cm, right=1.6cm]{geometry}
\usepackage{amsfonts, amsthm, comment, color}
\usepackage{amsmath, color,hyperref}
\usepackage{graphicx, mathrsfs}
\usepackage{cite}
\usepackage[normalem]{ulem}
\usepackage{tikz}
\usetikzlibrary{arrows,automata}
\usepackage{dsfont, sgame}
\usepackage{float}

\usepackage{tikz}
\usetikzlibrary{arrows}
\usetikzlibrary{decorations.markings}

\newcommand{\ev}[1]{\mathbb{E} \left [ #1 \right ]}

\newcommand{\pr}[1]{\mathbb{P} \left ( #1 \right )}
\newcommand{\prwrt}[2]{\mathbb{P}_{#1} \left [ #2 \right ]}

\newcommand{\one}[1]{\mathds{1}\left [ {#1} \right ]}
\newcommand{\onefunction}{\mathds{1}}

\newcommand{\vsup}{v_\text{sup}}
\newcommand{\vstrong}{v_\text{sup}^\text{s}}
\newcommand{\vweak}{v_\text{sup}^\text{w}}

\newcommand{\mA}{\mathcal{A}}
\newcommand{\mB}{\mathcal{B}}

\newcommand{\mS}{\mathcal{S}}
\newcommand{\mP}{\mathcal{P}}
\newcommand{\mX}{\mathcal{X}}
\newcommand{\mT}{\mathcal{T}}
\newcommand{\mE}{\mathcal{E}}

\newcommand{\mY}{\mathcal{Y}}

\newcommand{\gvalue}[1]{\mathbb{V}\left( #1 \right)} 
\newcommand{\gvalueA}[1]{\mathbb{V}_A\left( #1 \right)} 
\newcommand{\gvalueB}[1]{\mathbb{V}_B\left( #1 \right)} 
\newcommand{\osv}[1]{\mathfrak{u} \! \left ( #1 \right ) } 
\newcommand{\osvsymb}{\mathfrak{u} } 

\DeclareMathOperator*{\cav}{Cav}
\DeclareMathOperator*{\vex}{Vex}
\DeclareMathOperator*{\support}{Supp}

\newtheorem{lem}{Lemma}
\newtheorem{rem}{Remark}
\newtheorem{definition}{Definition}
\newtheorem{thm}{Theorem}

\newtheorem{prop}{Proposition}

\theoremstyle{definition}
\newtheorem{example}{Example}

\allowdisplaybreaks
\title{High Probability Guarantees in Repeated Games: \\ Theory and Applications in Information Theory}
\author{Payam~Delgosha$^*$, Amin~Gohari$^-$~and~Mohammad~Akbarpour$^\dagger$
\\\small$*$ Department of Electrical Engineering and Computer Sciences, University of California, Berkeley
\\\small pdelgosha@eecs.berkeley.edu
\\\small$-$ Department of Electrical Engineering, Sharif University of Technology
\\\small aminzadeh@sharif.edu
\\\small$\dagger$ Becker-Friedman Institute, University of Chicago \normalsize
\\\small akbarpour@uchicago.edu}
\date{}
\begin{document}
\maketitle

\begin{abstract}
We introduce a ``high probability" framework for repeated games with incomplete information. In our non-equilibrium setting, players aim to guarantee a certain payoff with high probability, rather than in expected value. 
We provide a high probability counterpart of the classical result of Mertens and Zamir for the zero-sum repeated games. Any payoff that can be guaranteed with high probability can be guaranteed in expectation, but the reverse is not true. Hence, unlike the average payoff case where the payoff guaranteed by each player is the negative of the payoff by the other player, the two guaranteed payoffs would differ in the high probability framework. 
One motivation for this framework comes from information transmission systems, where it is customary to formulate problems in terms of asymptotically vanishing probability of error. 
An application of our results to a class of compound arbitrarily varying channels is given.
\end{abstract}
\section{Introduction}

The standard game theory framework considers players who are von Neumann-Morgenstern (vNM) utility maximizers; that is, they maximize the expected value of some``utility function" defined over potential outcomes. The key to finding equilibria in such framework, of course, is to know the exact functional form of the utility function in order to translate payoffs and probabilities to utilities. The complexity of the analysis under non-standard functional forms, on the one hand, and the complications of identifying \emph{the} functional forms of the utilities of the real-world players, on the other hand, are two of the challenges of the standard framework.\footnote{That is probably one reason for that risk-neutrality of the players is an standard assumption in many games.}

In this paper, we undertake the above issues by introducing a non-equilibrium solution concept. We develop an analytical framework for (zero-sum) repeated games to study the following question: \emph{What is the highest payoff that players can ``guarantee with high probability?''} More precisely, we are concerned with payoffs that can be guaranteed (with some strategy) with probability $1-\epsilon$, where $\epsilon$ goes to $0$ as the games get played more and more. This ``high probability game theory" setting helps us to derive results analogous to the existing ones  on repeated games with incomplete information by Mertens and Zamir in \cite{RG-MZ71}.

\begin{figure}\begin{center}
\begin{game}{2}{2}[\textbf{Alice}][\textbf{Bob}]
    $\mathbf{s=0}$       & $U$ & $D$ \\
  $R$ & 1        & 2        \\
  $L$ & -4        & -6        
\end{game}
\quad
\begin{game}{2}{2}[\textbf{Alice}][\textbf{Bob}]
      $\mathbf{s=1}$     & $U$ & $D$ \\
  $R$ & 1        & 2        \\
  $L$ & 8        & 8        
\end{game}
\\
\begin{game}{2}{2}[\textbf{Alice}][\textbf{Bob}]
       & $U$ & $D$ \\
  $R$ & 1        & 2        \\
  $L$ & 2        & 1        
\end{game}
\end{center}
\caption[]{(Top)  payoff tables for Alice in state $s\in\{0,1\}$. (Bottom) the average table.}
\label{figcap1}
\end{figure}

Let us motivate our solution concept by a simple, concrete example. Consider the zero-sum repeated game depicted in Fig.~\ref{figcap1} between Alice and Bob. There is a state variable $S$ with uniform distribution over $\{0,1\}$. Alice's payoff table for $s=0$ and $s=1$ are given (Bob's payoff is negative of Alice's payoff). We assume that Alice and Bob have no knowledge of the value of $S$. The game is played  $n$ times between Alice and Bob, with the state variable being drawn at the beginning and kept fixed throughout the $n$ games. Alice and Bob only get to see their payoff values after playing all the $n$ games; hence they cannot gain any information about $S$ throughout the game. We make the assumption that if the total sum of the $n$ payoffs of the $n$ games of a player is positive, that player wins the entire game. There is a draw if the total sum of each player is zero. 

Let us first assume that Alice aims to maximize the expected value of her average payoffs in the $n$ games. Since Alice and Bob do not know $S$, we can compute the average table with weights $p_S(0)=p_S(1)=1/2$ as given in the bottom of Fig.~\ref{figcap1}.  The average table is symmetric and a Nash equilibrium strategy is for players to choose their actions uniformly at random. This gives Alice an expected average payoff of $1.5$. Thus, Alice can guarantee a positive total expected payoff  in the Nash equilibrium of the repeated game. However, with this strategy, Alice's average payoff is negative with probability $1/2$; it is $-7/4$ when $s=0$. Therefore, with probability $1/2$ when $s=0$, she will lose the entire $n$ game as her total sum payoff becomes negative with high probability by the law of large numbers. On the other hand,  assume that Alice plays a different strategy of choosing action $R$ all the time (which is not part of a Nash equilibrium). Then, Bob will play $U$ and this leads to a payoff of $1$ for Alice regardless of whether $s=0$ or $s=1$. The payoff of $1$ is smaller than the average payoff of $1.5$ that an equilibrium strategy will give her, but is guaranteed with probability one; thus ensuring that Alice will win the entire game.  

More generally, given an arbitrary repeated game (with complete or incomplete information), we ask that given an $\epsilon>0$, whether Alice has a strategy, for a sufficiently large enough $n$, that guarantees her total sum payoff to be greater than a number $v$ with probability $1-\epsilon$. In studying this natural problem, one may consider the whole $n$ games as a one stage strategic form game, and then consider the sequence of these games for different values of $n$, as $n$ becomes larger and larger. However, we find it easier to analyze this game as an extensive form repeated game in a high probability framework.

\textbf{Motivation from information theory:} 
One motivation for a high probability framework comes from information theory, where repeated use of a channel and a vanishing probability of error as the number of channel uses, $n$, tends to infinity is common. In the following we explain this via a simple example that requires little background in information theory. 

We need some definitions: a binary erasure channel (BEC) is a communication medium with a binary input $X\in \{0,1\}$. The output of this channel, denoted by random variable $Y$, is a symbol from $\{0,1,\mathsf e\}$ where $\mathsf e$ indicates that the input symbol is erased. When the input symbol is not erased $(Y\neq \mathsf{e})$, we have $Y=X$. The transmitter will not know whether a transmission has been erased at the receiver or not. 

 Let us denote the erasure event by random variable $E$, \emph{i.e.,} $E=1$ indicates that the input bit is \emph{not} erased. When we use the channel $n$ times, we will have erasure random variables $E_1, E_2, \cdots, E_n$ for each transmission. We assume that each $E_i\in\{0,1\}$ is a function of three variables: an internal channel state variable $S$, an input $A_i$ by Alice and an input $B_i$ by Bob, according to $E_i=g_S(A_i,B_i)$, where $g_s(a,b)$ is a given function for any $s\in\mathcal{S}$. Random variable $S$  is randomly chosen at the beginning and is fixed through the $n$ channel uses (slow fading). Alice and Bob have initial partial knowledge about $S$ by having access to $S_A$ and $S_B$ that are correlated with $S$. Figure~\ref{fig:BEC-repeated-game} illustrates this configuration. Alice aims to help the transmission (trying to make $E_i$ variables one, as much as possible) and Bob aims to disrupt it. Neither Alice nor Bob observe the variables $E_i$. But we assume that both Alice and Bob observe each other actions (inputs to the channel) causally; therefore, if they know each other's strategies, each party can infer some information about the other party's side information  by observing their actions. Hence, there is a tradeoff for both parties between using and hiding their side information: using it can be advantageous for the current transmission while actions can reveal information to the other party which could be turned against them in subsequent transmissions.

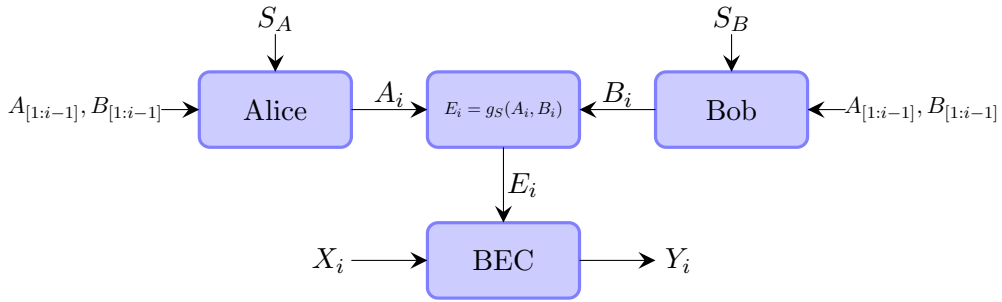
\begin{figure}
  \centering
  
  \begin{tikzpicture}
    \tikzstyle{box}=[rounded corners,fill=blue!20,draw=blue!50,very thick]
    \tikzstyle{bigarrow}=[decoration={markings,mark=at position 0.999 with
      {\arrow[scale=2]{stealth}}}, postaction={decorate}, shorten >=0.4pt]
    \draw[box] (-1,-0.5) rectangle (1,0.5);
    \node at (0,0) {BEC};
    \draw[bigarrow] (-2,0) -- (-1,0);
    \draw[box] (-1,1.5) rectangle (1,2.5);
    \node[scale=0.6] at (0,2) {$E_i=g_S(A_i, B_i)$};
    \draw[bigarrow] (0,1.5) -- (0,0.5);
    \node at (0.25,1) {$E_i$};
    
    \node at (-1.5,2.2) {$A_i$};
    \draw[box] (-4,1.5) rectangle (-2,2.5);
    \node at (-3,2) {Alice};
    \draw[bigarrow] (-2,2) -- (-1,2);
    
    \draw[bigarrow] (-3,3) -- (-3,2.5);
    \node at (-3,3.2) {$S_A$};
    \draw[bigarrow] (-4.5, 2) -- (-4,2);
    \node[scale=0.8] at (-5.5,2) {$A_{[1:i-1]}, B_{[1:i-1]}$};
    
    
    
    \draw[bigarrow] (1,0) -- (2,0);
    \node at (2.3, 0) {$Y_i$};
    \node at (-2.3, 0) {$X_i$};
    
    \draw[box] (4,1.5) rectangle (2,2.5);
    \node at (3,2) {Bob};
    \draw[bigarrow] (4.5, 2) -- (4,2);
    \node[scale=0.8] at (5.5,2) {$A_{[1:i-1]}, B_{[1:i-1]}$};
    \draw[bigarrow] (3,3) -- (3,2.5);
    \node at (3,3.2) {$S_B$};
    \node at (1.5,2.2) {$B_i$};
    \draw[bigarrow] (2,2) -- (1,2);
  \end{tikzpicture}
  \caption{An erasure channel where the erasure variable $E_i$ at time $i$ is produced in a repeated game}
  \label{fig:BEC-repeated-game}
\end{figure}

We can view the above as a game with incomplete information if we consider $E_i$ to be the payoff of the game for Alice (the payoff of Bob will be the negative of the payoff of Alice). Now, suppose that Alice can guarantee the expected total payoff of $n/2$. It may be the case that with probability $1/2$, her total payoff is zero, and with probability $1/2$ her total payoff is $n$. Then, with probability $1/2$ all the transmitted bits will be erased and no communication will be possible. Therefore, having a bound on the expected value of total payoff is not useful. On the other hand, given some small $\epsilon>0$, assume that Alice can guarantee her total payoff to be at least $n/3$ with probability $1-\epsilon$, regardless of how Bob plays. In other words, with probability $1-\epsilon$, at least $n/3$ bits from the $n$ bits that the transmitter sends will become available at the receiver. Then, with probability $1-\epsilon$, the transmitter can send about $n/3$ data bits by employing standard coding techniques such as fountain codes. Therefore, a high probability framework is of relevance to information transmission problem over this adversarial BEC channel.

It is possible to think of other information theory problems with a threshold phenomenon where the high probability framework is of relevance. For instance, in coding theory, the minimum distance of a code gives a guarantee that if the number of changes in a code sequence is sufficiently small, decoding will be successful. One can consider a problem where Alice and Bob are having actions that (along with a channel state) determines when a transmission will become erroneous. It would be desirable for Alice to make sure that the number of errors are bounded to ensure successful decoding. Or for instance, one can imagine a control system with two players, one who is trying to increase the error and the other who is trying to reduce the error. It may be that a bound on the total error of a system be of importance (and not its expected value). 

In Section \ref{compoundAVC}, we provide a technical application of the high probability framework for the problem of communication over a certain compound arbitrarily varying channel.  

\textbf{Our contribution:} In this paper, we focus on repeated games with incomplete information. Incomplete information refers to the fact that there are some unknown parameters that affect the payoff of the players. Each player has its own partial knowledge of the parameters, which may leak to the other player through actions during the repeated game. There is a tradeoff between 
hiding and using the information to each party. We refer the reader to \cite{RG-aumann} for a comprehensive treatment. Our main contribution in this paper is to find payoffs that can be guaranteed with high probability. We introduce a non-equilibrium approach -- the high probability condition -- and characterize payoffs that can satisfy that condition. Just like the average case framework, a complicating aspect of the problem is the tradeoff between 
hiding and using the information in the high probability framework. After proving our main result in the high probability framework, a non-trivial application of this framework to compound arbitrarily varying channel is also given.

There have been few previous works on implicit flow of information through actions in information theory \cite{Grover:EECS-2011-1, Cuff1}. However, none of existing works address implicit communication from the perspective of game theory to characterize the tradeoff between  hiding and using the information. Therefore, there are new conceptual features in our treatment. 

\textbf{Related work:} The literature of repeated game theory contains several ideas that are related to our paper. The standard approach to infinitely repeated games with no discount rate is the closest to ours, but it is concerned with the average payoff as a criteria of equilibrium \cite{MS06}. As we discussed before, our paper in some sense provides a high probability analogous of \cite{RG-MZ71}. Fudenberg and Levine \cite{FL92} study a repeated game of imperfect monitoring where they provide asymptotic bounds for the payoff of the player whose reputation (against his opponent) is crucial in identifying the equilibrium. The robust mechanism design literature is also related to ours, in that the goal is to ``guarantee'' a payoff (in a maxmin sense), but with a focus on single period games -- see, for instance, \cite{C15} and \cite{BM12}.

It should be pointed out that classical game theory has already found many applications in  information theory in scenarios where we have channels with unknown parameters, or channels that can vary arbitrarily (adversarial channels). The payoff function is generally either a mutual information (\emph{e.g.,} \cite{McEliece83,CsNa89,B57,  Dob61,  PDJL03}) or a coding error probability (\emph{e.g.,} \cite{BBT59,Root61}). Other than the problem of channels with uncertainty, game theory is 
vastly being used in other problems of information theory such as adversarial sources, power allocation and spectrum sharing.

\textbf{Organization:} 
The rest of this paper is organized as follows: in Section~\ref{notation} we define our notation. In Section~\ref{sec:game-definition-problem-statement} we formally define the problem, in Section~\ref{sec:expected-asymptotic} we review a result in repeated games with incomplete information in the expected value regime and finally in Section~\ref{sec:vs-vw} we prove our main result which is finding the highest value a party can guarantee with high probability in repeated games with incomplete information. Section \ref{compoundAVC} includes an application of the framework.

\subsection{Notation}
\label{notation}
We use capital letters for random 
variables and small letters for their realizations. We use $[i]$ to denote the set $\{1, 2, \dots, i\}$. Then $X_{[i]}$ denotes $X_1, X_2, \dots, X_i$. We use both subscript and superscripts to denote indicies; \emph{e.g.,} $X^j$ is rv $X$ indexed by $j$, and $X_i^j$ is rv $X$ indexed by $i$ and $j$. Thus, $X^{[k]}_{[n]}=\{(X_i^j: i\in [n], j\in[k]\}$. For a function $f$, $\cav f$ and $\vex f$ denote its lower concave envelope and upper convex envelope, respectively; \emph{e.g.,} $\cav f$ is the smallest concave function that lies above $f$. The support of a probability distribution $p$ over a finite set $\mathcal A$ is defined as $\support(p) = \{a \in \mathcal A: p(a) > 0 \}$.

\section{Definition and Problem Statement}
\label{sec:game-definition-problem-statement}

We consider a two player zero--sum game with Alice and Bob as players. We are interested in Alice's payoffs; hence Alice is the maximizer and Bob is the minimizer. 

\begin{definition}
  \label{def:value-of-game}
We define the value of a strategic game $\Upsilon$, $\gvalue{\Upsilon}$, as Alice's 
payoff in a Nash equilibrium; this value is the same for all Nash equilibriums since the game is zero--sum. We use $\gvalueA{\Upsilon}$ and $\gvalueB{\Upsilon}$ to denote Alice's and Bob's payoffs in any Nash equilibrium respectively. Hence, $\gvalue{\Upsilon} = \gvalueA{\Upsilon}=-\gvalueB{\Upsilon}$.
\end{definition}

A standard zero--sum repeated game of incomplete information consists of the following 
components \cite{RG-MZ71}:
\begin{itemize}
\item A zero--sum two player game $\Gamma$ called the \emph{stage game} which 
is repeated $n$ times. This game is between two players, say Alice and Bob, 
with finite sets of permissible actions $\mA$ and $\mB$, respectively. For each state 
$s \in \mS$, we have a payoff table $g_s$ where $g_s(a,b)$ denotes Alice's 
payoff when Alice plays action $a \in \mA$ and Bob plays action $b \in \mB$ in 
$\Gamma$.
\item A probability distribution $p_S(s)$ on a finite set of states, $\mS$, from 
which the state of the game is chosen by nature at random at the beginning of 
the game. Without loss of generality, we may assume that  $p_S(s)>0$ for all $s$, \emph{i.e.,} $\mS=\support(p)$.
\item This state is fixed throughout the $n$ repetitions of $\Gamma$, but 
neither Alice, nor Bob know the exact value of the state. Instead, Alice and 
Bob receive $S_A$ and $S_B$ as the side information about $S$, respectively. We 
assume that $S_A$ and $S_B$ are functions of $S$, \emph{i.e.,} $S_A = \mathscr T_A(S)$ and $S_B 
= \mathscr T_B(S)$. This assumption is made without loss of generality, as argued later. 
The alphabets of random variables $S_A$ and $S_B$ are denoted by $\mS_A$ and 
$\mS_B$, respectively.
\item Each party plays actions in the repeated game based on the information 
they have since the beginning of the game, \emph{i.e.,} their side informations 
$S_A$ and $S_B$ and the history of the game $A_{[i-1]}$ and $B_{[i-1]}$ 
which are Alice's and Bob's actions up to stage $i$ respectively. 
Note that in stage $k$, Alice and Bob play $A_k$ and $B_k$ simultaneously; here
we have shown 
actions with capital letters to emphasize that they are random variables since 
the two parties are allowed to employ random strategies, and the initial state $S$ is random.
\item We assume that Alice and Bob just observe their actions, not the payoffs 
they have received. When all $n$ stages are finished, Alice receives the time 
average of the payoffs of stage games, \emph{i.e.,}
\begin{equation}
\label{eq:sigma-n}
 \sigma_n := \frac{1}{n} \sum_{i=1}^n g_S(A_i,B_i).
\end{equation}
Note that $\sigma_n$ is a random variable.
\end{itemize}

The repeated game with above components is shown by
$
 \Gamma_n^{\mathscr T_A,\mathscr T_B}(p),
$
where $p$ is the prior distribution on state space $\mathcal{S}$. With an abuse of notation, we alternatively write
$\Gamma_n^{\mathscr T_A,\mathscr T_B}(S)$ where $S$ the random variable with distribution $p$.

A few points should be made about the above definition. First, note that the 
assumption that $S_A$ and $S_B$ are deterministic functions of $S$ is not 
restrictive. In fact, in the general case where $S_A$ and $S_B$ are allowed to 
be random functions of $S$, we can define a random variable $N$ where $S_A$ and 
$S_B$ are deterministic functions of $S$ and $N$ (functional representation lemma \cite[Appendix B]{ElGamalKim}). Therefore, for the new 
repeated game with state $\hat{S} = (S,N)$ and payoff tables $\hat{g}_{(s,n)} = 
g_s$ side informations are of our desired form and also the resulting payoffs 
do not change.

We can consider the strategic form for the above extensive form game and call 
it $\hat{\Gamma}_n^{\mathscr T_A,\mathscr T_B}(S)$. In this strategic form game, each action of a 
player 
is a pure strategy of him in the repeated game, \emph{i.e.,} a collection of 
deterministic functions determining what action should be played at each stage 
given the observations up to that time. The payoff of this game is the expected 
outcome of the repeated game defined as in \eqref{eq:sigma-n} when $S$ is 
generated from distribution $p_S(s)$. This strategic form game is indeed 
zero--sum, hence has a mixed strategy Nash equilibrium with value 
$\gvalue{\hat{\Gamma}_n^{\mathscr T_A,\mathscr T_B}(S)}$. This could be defined rigorously as 
follows:

\begin{definition}
 \label{def:strategic-form}
 The strategic form game $\hat{\Gamma}_n^{\mathscr T_A,\mathscr T_B}(S)$ is defined as a one 
stage zero--sum game with action sets $\hat{\mA}$ for Alice and $\hat{\mB}$ for 
Bob   where
\begin{equation}
 \begin{split}
  \hat{\mA} &= \{(f_1,\dots, f_n) \, | \, f_i: \mA_{[i-1]} \times \mB_{[i-1]} 
\times \mS_A \rightarrow \mA, 1 \leq i \leq n \} \\
\hat{\mB} &= \{(g_1,\dots, g_n) \, | \, g_i: \mA_{[i-1]} \times \mB_{[i-1]} 
\times 
\mS_B \rightarrow \mB, 1 \leq i \leq n \} \\
 \end{split}
\end{equation}
where $f_i(a_{[i-1]}, b_{[i-1]}, s_A)$ determines which action Alice will 
play if the history of the game is $a_{[i-1]}, b_{[i-1]}$ and she has the 
side information $S_A$, and Bob's strategies are similar. Given a realization $S = s$, a unique 
deterministic sequence of actions is played by Alice and Bob, denoted by 
$a_{[n]}(s), b_{[n]}(s)$ where
\begin{equation}
\begin{split}
 a_i(s) &= f_i(a_{[i-1]}(s), b_{[i-1]}(s), \mathscr T_A(s) )  \\
 b_i(s) &= g_i(a_{[i-1]}(s), b_{[i-1]}(s), \mathscr T_B(s) ).  \\
 \end{split}
\end{equation}
The payoff function 
of this game is defined as follows: 
\begin{equation}
\frac{1}{n} \sum_s p_S(s) \sum_{i=1}^n g_s(a_i(s), b_i(s)).
\end{equation}
\end{definition}

As mentioned above, this is a finite zero--sum game, hence has a mixed strategy 
Nash equilibrium. Any strategy of this form is a mixture of pure strategies 
defined above, called a mixed strategy in the repeated game. However, since the 
repeated game $\Gamma_n$ is with perfect recall, \emph{i.e.,}\ each player remembers 
his own past actions, Kuhn's theorem implies that without loss of generality we 
may only consider behavioral strategies (see, \cite{Osborne}, for instance). A behavioral strategy is  a collection 
of random functions assigning probabilities to each action given the history of 
the game at each stage:

\begin{definition}
A behavioral strategy of Alice in the game $\Gamma_n$ is a collection of random 
functions where
\begin{equation}
 \alpha(a_i | a_{[i-1]}, b_{[i-1]}, s_A),
\end{equation}
is the probability that Alice chooses action $a_i$ when the history of the 
game is $a_{[i-1]}, b_{[i-1]}$ and Alice's side information is $s_A$. Bob's behavioral strategies are defined similarly via $\beta(b_i | a_{[i-1]}, b_{[i-1]}, s_B)$. The 
choices of Alice and Bob in different stages are assumed to be conditionally independent given the past action history\footnote{In other words, the players do not use private randomization to make further correlation in their actions}, \emph{i.e.,} the probability distribution on the outcome of the 
game is
\begin{equation}
 p(s,a_{[n]},b_{[n]}) = p_S(s) \prod_{i=1}^n \alpha(a_i | a_{[i-1]}, b_{[i-1]}, 
\mathscr T_A(s)) \beta(b_i | a_{[i-1]}, b_{[i-1]}, 
\mathscr T_B(s)).
\end{equation}

The set of Alice's behavioral strategies in $\Gamma_n$ is denoted by 
$\tilde{\mA}_n$ and Bob's behavioral strategies is denoted by $\tilde{\mB}_n$.
\end{definition}

The value of $\Gamma_n$ is defined as the value of its strategic form. As a result of Kuhn's theorem, we have
\begin{equation}
\label{eq:repeated-equilibrium}
 \gvalue{\Gamma_n^{\mathscr T_A,\mathscr T_B}(S) } = \max_{\alpha \in \tilde{\mA}_n} \min_{\beta 
\in \tilde{\mB}_n} \ev{ \frac{1}{n}\sum_{i=1}^n g_S(A_i, B_i) } = \min_{\beta 
\in \tilde{\mB}_n} \max_{\alpha \in \tilde{\mA}_n} \ev{ \frac{1}{n} \sum_{i=1}^ng_S(A_i, 
B_i) },
\end{equation}
where $A_i$ and $B_i$ are random variables denoting the actions of Alice and Bob.

Let $
 \sigma_n = \frac{1}{n} \sum_{i=1}^n g_S(A_i,B_i).
$ be 
the time average payoff of Alice. Then, equation \eqref{eq:repeated-equilibrium} implies that if Alice plays her 
equilibrium strategy, independent of Bob's strategy, we have
\begin{equation}
 \ev{\sigma_n} \geq \gvalue{\Gamma_n^{\mathscr T_A, \mathscr T_B}(S)},
\end{equation}
which shows that Alice can guarantee $\gvalue{\Gamma_n^{\mathscr T_A,\mathscr T_B}(S)}$ in 
the average sense by playing an equilibrium (behavioral) strategy. Conversely, from \eqref{eq:repeated-equilibrium}, if Bob plays his equilibrium 
strategy, Alice can not guarantee more than the value of the game, \emph{i.e.,}
$\ev{\sigma_n} \leq \gvalue{\Gamma_n^{\mathscr T_A,\mathscr T_B}(S)}$. Hence
$\gvalue{\Gamma_n^{\mathscr T_A,\mathscr T_B}(S)}$ is the maximum value Alice can guarantee in 
the expected value sense. The asymptotic behavior of this value, \emph{i.e.,} $\lim_{n 
\rightarrow \infty} \gvalue { \Gamma_n^{\mathscr T_A,\mathscr T_B}(S) }$ is analyzed by Mertens 
and Zamir in \cite{RG-MZ71}. We will review a special case of this result in 
Section~\ref{sec:expected-asymptotic}.

On the other hand, one might be interested in finding the value Alice can 
guarantee with high probability instead of in average. There are two 
ways of defining this concept.

\begin{definition}
 \label{def:strong-guaranteeing}
 We say that Alice can strongly guarantee a value $v$ if for all $\epsilon>0$, 
there exists a natural number $N$ such that for all $n > N$, Alice has a strategy $\alpha$ in 
$\Gamma_n^{\mathscr T_A,\mathscr T_B}(p_S)$ so that for all strategies $\beta$ of Bob in this 
game we have
\begin{equation}
 \pr{\sigma_n < v} < \epsilon.
\end{equation}
\end{definition}

\begin{definition}
\label{def:weak-guaranteeing}
We say that Alice can weakly guarantee a value $v$ if for all $\epsilon>0$, there 
exists $N$ such that for all $n > N$ and for all strategy $\beta$ for Bob in 
$\Gamma_n^{\mathscr T_A,\mathscr T_B}(p_S)$, there exists a strategy $\alpha$ for Alice in this 
game such that
\begin{equation}
 \pr{\sigma_n < v} < \epsilon.
\end{equation}
\end{definition}

Note that the difference between the above two definitions is that if Alice 
wants to guarantee a payoff strongly, then she needs to have a \emph{universal} 
strategy $\alpha$ independent of Bob's strategy. A universal strategy of Alice should work for all 
possible strategy of Bob. On the other hand, when Alice wants to guarantee a value weakly, 
she can adapt her strategy based on Bob's strategy. Therefore, it is evident 
that if Alice can guarantee a value in the strong sense, she can guarantee it 
in the weak sense too.

\begin{definition}
 \label{def:vsup}
When the game state has distribution $p_S$, Alice's and Bob's side information functions are $\mathscr T_A$ and $\mathscr T_B$, respectively, 
 we denote the supremum of all values Alice can strongly guarantee as 
$\vstrong(p_S, \mathscr T_A, \mathscr T_B)$. Similarly $\vweak(p_S, \mathscr T_A, \mathscr T_B)$ denotes the supremum over all values Alice can 
guarantee weakly. When it is clear from the context, we use $\vstrong$ and $\vweak$ instead as shorthands for $\vweak(p_S, \mathscr T_A, \mathscr T_B)$ and $\vstrong(p_S, \mathscr T_A, \mathscr T_B)$, respectively.
\end{definition}

We will find the values of $\vstrong$ and $\vweak$ in Section~\ref{sec:vs-vw}.

\section{Review of results for the expected value payoff regime}
\label{sec:expected-asymptotic}

In this section, we review an existing result for guaranteeing payoffs in the 
expected value. In this approach,  the 
Nash Equilibrium of the $n$ stage game, $\gvalue{\Gamma_n}$ is asymptotically 
analyzed and its limit value as well as its 
convergence rate is obtained. 

We first need a definition:
\begin{definition}
\label{def:u-average-game}
Given a distribution $p_S$ on set $\mathcal S$ and payoff tables $g_s(a,b)$ for $s\in \mathcal{S}$,  
define $\osv{p_S}$ as the value of the one-stage zero-sum game with the average payoff table $\sum_s p_S(s) 
g_s$. We may also denote it by $\osv{S}$ where $S$ is the random variable 
with distribution $p_S$.
\end{definition}

Consider the special case where one player is fully aware of the game state and the other has no side information. In order to do so we employ the notation $\emptyset$ as the function which gives no side information, \emph{i.e.,} it has a constant output $\emptyset(s)=0$  for all $s \in \mS$. On the other hand, let $\onefunction$ is the side information function which gives full information, \emph{i.e.,} $\onefunction(s) = s$ for all $s \in \mS$. We consider the case where $\mathscr T_A=\emptyset, \mathscr T_B=\onefunction$. Then,

\begin{thm}[Theorem~3.16 in \cite{RG-Zamir1992}]
  \label{thm:zamir-side-information-on-one-side}
  $\lim_{n \rightarrow \infty} \gvalue{\Gamma_n^{\emptyset, \onefunction}(p_S)}$ exists and is equal to $\vex \osv{p_S}$ where $\vex \osvsymb$ is the convex hull of $\osvsymb$ as a function on the probability simplex. Furthermore there exists a constant $C$ such that for all $p_S$ we have 
  \begin{equation}
    \label{eq:incomplete-information-one-side-uniform-rate}
    0 \leq \vex \osv{p_S} - \gvalue{\Gamma_n^{\emptyset, \onefunction}(p_S)} \leq \frac{C}{\sqrt{n}}.
  \end{equation}
\end{thm}

\begin{rem}
  In \cite{RG-Zamir1992}, Alice is assumed to have full information and Bob knows nothing; in fact their place is reversed. In order to change their place, we can negate the payoff table. That is why we have $\vex$ instead of $\cav$ here and also the inequality direction in~\eqref{eq:incomplete-information-one-side-uniform-rate} is reversed. 
To be more precise, statement of Theorem~3.16 of \cite{RG-Zamir1992} in our notation translates to
\begin{equation*}
  0 \leq \gvalueB{\Gamma_n^{\emptyset, \onefunction}(p_S)} - \cav( - \osv{p_S}) \leq \frac{C}{\sqrt{n}}.
\end{equation*}
Noting $\gvalueA{\Upsilon} = - \gvalueB{\Upsilon}$ for any zero sum game $\Upsilon$ and $\cav(-f) = - \vex(f)$ for any function $f$ transforms the above equation into \eqref{eq:incomplete-information-one-side-uniform-rate}.
Also note that on the right hand side of the analogue of \eqref{eq:incomplete-information-one-side-uniform-rate} in \cite{RG-Zamir1992} we have the term  $\sum_{s \in \mS} \sqrt{p_S(s) (1-p_S(s))}$ which is upper bounded by $|\mS|$ and is absorbed into the constant $C$ here. 
\\
  Observe that the constant $C$ in \eqref{eq:incomplete-information-one-side-uniform-rate} is independent of $p_S$, hence it implies uniform convergence of the sequence $\gvalue{\Gamma_n^{\emptyset, \onefunction}(p_S)}$ to its limit on $p_S$. 
\end{rem}

In the following, we provide an intuitive sketch of the key ideas used to prove 
Theorem~\ref{thm:zamir-side-information-on-one-side}; see \cite{RG-Zamir1992} for a rigorous proof. Alice initially does not know anything about $S$. Bob knows $S$ and his actions may increase Alice's information about $S$. Let us denote Alice's information about $S$ at time stage $i$ by the mutual information $J_i=I(S;A_{[i-1]}B_{[i-1]})$ for $i\in[n]$. The sequence $\{J_i\}$ satisfies the following properties: $J_1=0$, $J_i\leq J_{i+1}$ and $J_i\in [0, H(S)]$. Take some $\delta>0$. We say that an information jump occurs at stage $i$ if $J_{i}-J_{i-1}\geq \delta$. Since $J_i\in [0, H(S)]$, the number of jumps is at most the  constant $k=H(S)/\delta$. Let $\mathcal{I}=\{i\in [n]: J_{i}-J_{i-1}\leq \delta\}$. Since $k$ is a constant, $|\mathcal I|\geq n-k$. The payoff of Alice is its average over time stages $1$ to $n$ and is dominated by the average of stages in $\mathcal I$, \emph{i.e.,} 
$$\frac{1}{n} \sum_{i=1}^n g_S(A_i,B_i) \approx  \frac{1}{|\mathcal I|} \sum_{i\in\mathcal I} g_S(A_i,B_i).$$
At time instances in $i\in \mathcal I$, Bob's strategy is essentially non-revealing in the sense that if from Alice's view, $S$ has conditional pmf $q_i(s)=p(s|a_{[i-1]}b_{[i-1]})$ at time stage $i$, we have that $q_{i}(s)\approx q_{i+1}(s)$. Then, the payoff that Alice can obtain at time stage $i$ is that of a non-revealing $\osv{q_i(s)}$. The average payoff over various realizations of $a_{[i-1]}b_{[i-1]}$ is equal to $$\sum_{a_{[i-1]}b_{[i-1]}}p(a_{[i-1]}b_{[i-1]})\osv{p(s|a_{[i-1]}b_{[i-1]})}\geq  \vex \osv{p}$$
as $\sum_{a_{[i-1]}b_{[i-1]}}p(a_{[i-1]}b_{[i-1]})p(s|a_{[i-1]}b_{[i-1]})=p(s)$. This demonstrates that Alice's payoff is greater than or equal to $ \vex \osv{p}$, regardless of how Bob plays.

On the other hand, Bob has a strategy ensuring that Alice's payoff does not exceed $\vex \osv{p}$. Assume that
$$\vex \osv{p}=\sum_{i=1}^k\lambda_i\osv{p_i(s)}$$
for some non-negative weights $\lambda_i, i\in[k]$ adding up to one, and pmfs $p_i(s)$ satisfying $$\sum_i\lambda_ip_i(s)=p(s).$$
Let $V$ be a random variable on alphabet set $\{1,2,\cdots, k\}$ satisfying $p(V=i)=\lambda_i$. Rv $V$ is joint distributed with $S$ as follows:
$$p(V=i, S=s)=\lambda_i p_i(s).$$
Bob can locally create $V$ by passing $S$ through a channel $p(v|s)$. Bob's strategy is then as follows: he uses his actions in the  first few instances of the game to communicate $V$ to Alice. The payoff in these first few instances of the game do not affect the overall payoff over the $n$ games. By doing this, Bob is effectively announcing $V$ to Alice, at no effective cost. Bob then proceeds as follows: he completely forgets the exact state $S$ and only given the variable $V$, he plays the optimal strategy of $\osv{p_i}$ when $V=i$. In this case, since the marginal distribution of $S$ is $p_i$ and Alice knows whatever Bob knows about the state, the posterior of the state does not change from stage to stage from Alice's point of view,\emph{i.e.,} she does not learn further about the state from Bob's actions than the initial announcement $V$. Hence, 
\begin{equation*}
  \ev{\frac{1}{n} \sum_{i=1}^n g_S(A_i, B_i)} = \sum_{i=1}^k \lambda_i \frac{1}{n} \sum_{j=1}^n \ev{g_S(A_j, B_j) | V= i} \leq \sum_{i=1}^k \lambda_i \osv{p_i}=\vex \osv{p}.
\end{equation*}
Roughly speaking, this argument shows that the optimal strategy for the informed player is to announce whatever the uninformed player is eventually going to learn about the state at the beginning of the game and forget the extra information, so that both players end up having a balanced information about the state. This completes the sketch of the proof of  \cite{RG-Zamir1992}.

An interesting implication of Theorem \ref{thm:zamir-side-information-on-one-side} is as follows: considering the mixed Nash strategies, Alice's mixed strategy ensures learning and exploiting from Bob's actions about state $S$ in an optimal way, for all possible strategies of Bob. In other words, it implies existence of a ``universal" algorithm for Alice that performs as if Alice knew Bob's strategy. 

\section{Guaranteeing with High Probability}
 \label{sec:vs-vw}

In this section we find the values of $\vstrong$ and $\vweak$. Without loss of generality, we  assume that $p_S(s), p_{S_A}(s_A)>0$ for all $s\in\mathcal S, s_A\in\mathcal S_A$, where $S_A=\mathscr T_A(S)$. Therefore $\mathscr T_A^{-1}(s_A) := \{s \in \mS: \mathscr T_A(s) = s_A\}$ is non-empty for all $s_A$. Our main result is the following:
\begin{thm}
 \label{thm:main-two-sided}
 We have
 \begin{equation}
  \vstrong(p_S, \mathscr T_A, \mathscr T_B) = \vweak(p_S, \mathscr T_A, \mathscr T_B) = \min_{s_A} \min_{p_S: \support(p_S) \subseteq \mathscr T_A^{-1}(s_A)} \osv{S}.
 \end{equation}
\end{thm}

\begin{example}
Consider the game tables given in Figure~\ref{figcap2} where the numbers in the table are Alice's payoff.
\begin{figure}\begin{center}
\begin{game}{2}{2}[\textbf{Alice}][\textbf{Bob}]
    $\mathbf{s=0}$       & $U$ & $D$ \\
  $R$ & -1        & 0        \\
  $L$ & 0        & 0       
\end{game}
\quad
\begin{game}{2}{2}[\textbf{Alice}][\textbf{Bob}]
      $\mathbf{s=1}$     & $U$ & $D$ \\
  $R$ & 0        & 0        \\
  $L$ & 0        & -1        
\end{game}
\\
\begin{game}{2}{2}[\textbf{Alice}][\textbf{Bob}]
       & $U$ & $D$ \\
  $R$ & -p        & 0        \\
  $L$ & 0        & -(1-p)        
\end{game}
\end{center}
\caption[]{(Top)  payoff tables for Alice in state $s\in\{0,1\}$. (Bottom) the average table.}
\label{figcap2}
\end{figure}
Assume that Bob knows the exact value of $S$, while Alice has no side information about $S$. The average table is also given in the figure. One can easily obtain $\osv{p} = - p (1-p)$ (\cite[Sec.~3.2.5.]{RG-Zamir1992}). Since $\osvsymb(\cdot)$ is convex, the maximum value that Alice can guarantee in expected value is $-p(1-p)$.
However, since Alice has no side information, we get
\begin{equation*}
  \vstrong = \vweak = \min_{p} \osv{p} = -\frac{1}{4},
\end{equation*}
which is strictly less than the expected value case unless $p = 1/2$. A naive approach suggests that perhaps it is more beneficial for Bob to play $U$ if $s=0$, and play $D$ if $s=1$. However, note that in this case, Alice after observing Bob's actions realizes the true state and plays  $L$ for $s=0$, and  $R$ for $s=1$. While if Bob chooses each column with probability $1/2$ independent of the state (which is a completely non--revealing strategy), then Alice does not gain any information about the true state and should choose one row with probability half (since she does not know where the $-1$ is located). This would guaruntee her a payoff of $-1/4$ in high probability. On the other hand, for the expected payoff regime, the optimal average payoff of Alice is $\vex \osv{p}=\osvsymb(p)$, and this is obtained by Bob playing the equilibrium strategy of the average table without using his knowledge of the state.
\end{example}

Before getting into the proof of this theorem in Section \ref{sec-sub-l2}, we prove a few lemmas. Our first observation is that 
the values of $\vweak$ and 
$\vstrong$ depend only on the support of $p(s)$.
\begin{lem}
\label{lem:support-invariant}
Assume $p_S$ and $\tilde{p}_S$ are two distributions on $\mS$ such that $\support p_S = \support \tilde{p}_S$. Then we have 
\begin{align*}
  \vstrong(p_S, \mathscr T_A, \mathscr T_B) &= \vstrong(\tilde{p}_S, \mathscr T_A, \mathscr T_B),
\\
  \vweak(p_S, \mathscr T_A, \mathscr T_B) &= \vweak(\tilde{p}_S, \mathscr T_A, \mathscr T_B).
\end{align*}
\end{lem}
\begin{proof}
Note that
$
 \pr{\sigma_n <v} = \sum_s p_S(s) \pr{\sigma_n < n|s}
$. 
Therefore, if $\pr{\sigma_n < v} < \epsilon$, then we have
\begin{equation}
 \pr{\sigma_n < v | s} < \frac{\epsilon}{p_\text{min}}, \quad \forall s\in\mathcal S,
\end{equation}
where $p_\text{min}\neq 0$ is the minimum value of $p(s)$ on its support. Then, we have
\begin{equation}
 \prwrt{\tilde{p}}{\sigma_n < v} = \sum_s \tilde{p}(s) \pr{\sigma_n < v|s} \leq 
\sum_s \tilde{p}(s) \frac{\epsilon}{p_\text{min}} = 
\frac{\epsilon}{p_\text{min}},
\end{equation}
which could be made small enough by setting $\epsilon$ sufficiently small. 
\end{proof}

\begin{rem}
\label{rem:vweak-vstrong-set-notation}
As a result of this lemma,  for a subset $\mS' \subseteq \mS$ we may use $\vweak(\mS', \mathscr T_A, \mathscr T_B)$ and $\vstrong(\mS', \mathscr T_A, \mathscr T_B)$ as the value of $\vweak(q_S, \mathscr T_A, \mathscr T_B)$ and $\vstrong(q_S, \mathscr T_A, \mathscr T_B)$, respectively, for any distribution $q_S$ with $\support q_S = \mS'$. In fact, $\vweak(\mS', \mathscr T_A, \mathscr T_B)$ and $\vstrong(\mS', \mathscr T_A, \mathscr T_B)$ could be interpreted as values that Alice can guarantee ``for each possible state in $\mS'$'' in the worst case regime. 
\end{rem}

In the following lemma, we reduce the problem of finding $\vweak$ and $\vsup$ to the case where Alice has zero side information about the game state and Bob exactly knows its value. We use the notations $\emptyset$ and $\onefunction$ from the previous section.

\begin{lem}
\label{lem:reduction-one-side-information}
We have
\begin{align}
\label{eq:vweak-empty-one}
  \vweak(\mS, \mathscr T_A, \mathscr T_B) &= \min_{s_A} \vweak(\mathscr T_A^{-1}(s_A), \emptyset, \onefunction),
\\
  \label{eq:vstrong-empty-one}
  \vstrong(\mS, \mathscr T_A, \mathscr T_B) &= \min_{s_A} \vstrong(\mathscr T_A^{-1}(s_A), \emptyset, \onefunction).
\end{align}
\end{lem}
\begin{proof}
We first show that \begin{equation}\vweak(\mS, \mathscr T_A, \mathscr T_B) =\vweak(\mS, \mathscr T_A, \onefunction)\label{eqn:tbto1}\end{equation} and similarly for $ \vstrong$. In other words, $\vweak$ does not depend on $\mathscr T_B$ and from Alice's perspective, it is always as if Bob knows the state perfectly. 
To show this, consider the following strategy for Bob: he guesses  the state $S$ randomly and proceeds assuming that his guess is the correct value for $S$. Since the state space is finite, with a nonzero and constant probability his guess becomes true. But since Alice should guarantee with high probability, she can not neglect the constant probability of Bob's guess becoming true. Therefore, her strategy should be for the worst case, guaranteeing her payoff conditioned on the event that Bob's guess about the state is correct. This completes the proof for $\vweak(\mS, \mathscr T_A, \mathscr T_B) =\vweak(\mS, \mathscr T_A, \onefunction) $.

It remains to show that
\begin{equation*}
  \vweak(\mS, \mathscr T_A, \onefunction) = \min_{s_A} \vweak(\mathscr T_A^{-1}(s_A), \emptyset, \onefunction),
\end{equation*}
 and similarly for $ \vstrong$. When Alice receives a side information $s_A$, any of the states in the set $\mathscr T_A^{-1}(s_A)$ may have happened.  Since Alice has no further initial side information other than $s_A$, we can assume that state space is reduced to $\mathscr T_A^{-1}(s_A)$ with Alice having zero side information. Then, $\vweak(\mathscr T_A^{-1}(s_A), \emptyset, \onefunction)$ would be the payoff that can be guaranteed in this case. Since Alice should guarantee for any possible value of $s_A$, the maximum payoff she can guarantee is $\min_{s_A} \vweak(\mathscr T_A^{-1}(s_A), \emptyset, \onefunction)$.
\end{proof}

\subsection{Proof of Theorem 
 \ref{thm:main-two-sided}}\label{sec-sub-l2} By Lemma \ref{lem:reduction-one-side-information}, we only need to show that 
 \begin{equation}
  \vstrong(\mS, \emptyset, \onefunction) = \vweak(\mS, \emptyset, \onefunction) =  \min_{p_S} \osv{p_S}.
 \end{equation}
Since $\vstrong \leq \vweak$, it suffices to show the following two propositions:
\begin{prop}
 \label{prop:lower-bound-final}
 We have
\begin{equation}
  \vstrong(\mS, \emptyset, \onefunction)\geq \min_{p_S} \osv{p_S}.\label{eqn:step2}
\end{equation}
\end{prop}
 \begin{prop}
 \label{prop:uppernound}
 We have
 \begin{equation}
  \vweak(\mS, \emptyset, \onefunction) \leq \min_{p_S} \osv{p_S},\label{eqn:step1}
 \end{equation}
\end{prop}
To prove the above propositions, we first show a lemma:
\begin{lem}
 \label{prop:lower-bound-sasl}
 We have 
 \begin{equation}
  \vstrong(\mS, \emptyset, \onefunction) \geq \gvalue{\Omega_n(\mS)} = \min_{p(s)} 
\gvalue{\Gamma^{\emptyset,\onefunction}_n(p)}, \qquad \forall n \in \mathbb{N}\label{eqn22par1}
 \end{equation}
where ${\Omega}_n(\mS)$ is an auxiliary zero--sum game in which Bob chooses state $s$ (the table $g_s$) from the set $\mS$
once and for all, and 
Alice receives no side information, and then each 
player 
observes the history of the game (expect that Alice does not observe Bob's 
action on choosing the table). The game is played for $n$ stages. The final payoff of Alice is the average of her payoff in the $n$ subgames, according to the payoff table $g_s$ with $s$ chosen by Bob in his first action.
\end{lem}

\begin{proof}[Proof of Lemma \ref{prop:lower-bound-sasl}] 
Note that ${\Omega}_n(\mS)$ is a repeated zero--sum game with 
perfect recall, so using Kuhn's Theorem, we may consider behavioral 
strategies in a Nash equilibrium of this game.

Assume $v = \gvalue{{\Omega}_n(\mS)}$ is the value of 
${\Omega}_n(\mS)$ and $\tilde{\alpha}$ be an equilibrium strategy for 
Alice. This means that for all strategy $\tilde{\beta}$ for Bob, the expected 
value of Alice by playing $\tilde{\alpha}$ is at least $v$.

Now, we repeat game ${\Omega}_n(\mS)$, $m$ times. Hence, we have a 
game of size $mn$ with $m$ blocks of length $n$. At the beginning of each block, a new value for $s$ (a new payoff 
table) is chosen by Bob and the game of length $n$ is played. We call the state 
of block $i$ as $S_i$ and actions of this block by 
$a^i_{[n]}$ and $b^i_{[n]}$ for Alice and Bob, respectively. Here $a^i_j$ for $i\in[m], j\in[n]$ is the $j$-th action of Alice in the block $i$.

Assume Alice plays strategy $\tilde{\alpha}$ in an i.i.d.\ fashion in each block, which means 
that she plays action $a^i_j$ at block $i$ with probability
\begin{equation}
 \alpha\left (a^i_j| a^{[i-1]}_{[n]} a^i_{[j-1]} 
b^{[i-1]}_{[n]} b^i_{[j-1]}\right ) = \tilde{\alpha} \left (a^i_j | 
 a^i_{[j-1]} b^i_{[j-1]} \right).
\end{equation}
Now we claim that playing this strategy by Alice results in guaranteeing $v - 
\epsilon$ with high probability for her when $m$ is large enough. For doing so, assume that Bob plays an 
arbitrary strategy in the game with length $mn$. More precisely he chooses 
state $s_i$ for block $i$ with probability
\begin{equation}
 \beta\left (s_i | s_{[i-1]} a^{[i-1]}_{[n]} b^{[i-1]}_{[n]} \right),
\end{equation}
and action $b^i_j$ with probability
\begin{equation}
 \beta\left (b^i_j | s_{[i]} a^{[i-1]}_{[n]} a^i_{[j-1]} 
b^{[i-1]}_{[n]} b^i_{[j-1]} \right ).
\end{equation}
Now define the random variable $W_k$ to be 
\begin{equation}
 W_k = \sum_{i=1}^k \sum_{j=1}^n g_{S_i} (A^i_j, B^i_j) - n k v,
\end{equation}
which is the sum of the payoffs of Alice in 
the first $k$ blocks, centered by the expected payoff. Now we claim that $W_k$ 
is submartingale with respect to $A^{[k]}_{[n]}, B^{[k]}_{[n]}, 
S_{[k]}$. Note that 
\begin{equation}
\label{eq:w-k+1-w-k-submartingale}
 \ev{W_{k+1} \bigg | A^{[k]}_{[n]}, B^{[k]}_{[n]}, 
S_{[k]}} = W_k + \ev{\sum_{j=1}^n g_{S_{k+1}} (A^{k+1}_j, B^{k+1}_j) \Bigg | 
A^{[k]}_{[n]}, B^{[k]}_{[n]}, 
S_{[k]}} - nv.
\end{equation}
Now we claim that
\begin{equation}
 \ev{\sum_{j=1}^n g_{S_{k+1}} (A^{k+1}_j, B^{k+1}_j) \Bigg | 
A^{[k]}_{[n]}, B^{[k]}_{[n]}, 
S_{[k]}} \geq nv.
\end{equation}
It suffices to show that for any realization of the history, $s_{[k]}, 
a^{[k]}_{[n]} b^{[k]}_{[n]}$, the expected value is at least $nv$. To show this, note 
that for this specific realization of the history, the term inside the 
expectation is the sum of Alice's payoff in a game ${\Omega}_n$ 
where Alice uses equilibrium strategy $\tilde{\alpha}$ and Bob uses strategy
\begin{equation}
 \tilde{\beta} (s_{k+1}) = \beta\left (s_{k+1} \bigg | s_{[k]} 
a^{[k]}_{[n]} b^{[k]}_{[n]} \right ),
\end{equation}
and
\begin{equation}
\tilde{\beta} (b^{k+1}_j | s_{k+1}, a^{k+1}_{[j-1]} b^{k+1}_{[j-1]}) = 
\beta \left (b^{k+1}_j \bigg | s_{[k+1]} a^{[k]}_{[n]} a^{k+1}_{[j-1]} 
b^{[k]}_{[n]} b^{k+1}_{[j-1]}  \right ).
\end{equation}
Since $\tilde{\alpha}$ is an equilibrium strategy, for all strategy of Bob 
including the above $\tilde{\beta}$ in block $k+1$ the expected value of 
Alice's payoff is at least the value of the game. Hence
\begin{equation}
 \ev{\sum_{j=1}^n g_{S_{k+1}} (A^{k+1}_j, B^{k+1}_j) \Bigg | 
a^{[k]}_{[n]}, b^{[k]}_{[n]}, 
s_{[k]}} \geq nv \qquad \forall a^{[k]}_{[n]}, b^{[k]}_{[n]}, 
s_{[k]}.
\end{equation}
Therefore
\begin{equation}
 \ev{\sum_{j=1}^n g_{S_{k+1}} (A^{k+1}_j, B^{k+1}_j) \Bigg | 
A^{[k]}_{[n]}, B^{[k]}_{[n]},
S_{[k]}} \geq nv,
\end{equation}
Substituting this into \eqref{eq:w-k+1-w-k-submartingale} 
shows that $W_k$ is a submartingale.

Note that
\begin{equation}
 | W_{k+1} - W_k | = | \sum_{j=1}^n g_{S_{k+1}}(A^{k+1}_j, B^{k+1}_j) - nv| 
\leq 2nM,
\end{equation}
where $M$ is an upper bound on payoffs. Now using Azuma's inequality with $W_0 
= 0$ we have
\begin{equation}
 \pr{W_m  < -t} \leq \exp \left( \frac{-t^2}{2m(2Mn)^2} \right).
\end{equation}
Setting $t = m^\delta$ for a $1/2 < \delta < 1$, the above bound goes to zero 
with $m$ going to infinity. Therefore for $m$ large enough, with high 
probability we have $W_m \geq - m^\delta$ or equivalently
\begin{equation}
 \sum_{k=1}^m \sum_{j=1}^n g_{S_k}(A^k_j, B^k_j) \geq nmv - m^\delta,
\end{equation}
or
\begin{equation}
 \frac{1}{nm} \sum_{k=1}^m \sum_{j=1}^n g_{S_k}(A^k_j, B^k_j) \geq v - 
\frac{m^{\delta-1}}{n} \geq v - \epsilon,
\end{equation}
where the last inequality holds with high probability for $m$ large enough. Therefore, Alice can 
guarantee payoff $v$ with high probability for the game with the game 
${\Omega}_n$ repeated $m$ times by playing $\tilde{\alpha}$ 
i.i.d.

Next, observe that playing the same strategy by Alice can guarantee her payoff 
$v-\epsilon$ for game ${\Omega}_{nm}$ for large enough $m$. The reason is that Bob's strategies in ${\Omega}_{nm}$ is a subset of Bob's strategies 
in the $m$ repetition of ${\Omega}_n$, as in the former Bob chooses $s$ once at the beginning while in the latter, he is allowed to choose it at the beginning of each of the $m$ blocks. Finally, observe that Alice can guarantee payoffs arbitrarily close to
$v$ for game ${\Omega}_{k}$, as long as $k$ is large enough, even when $k$ is not of the product form $nm$ for some $m$. Let $k=mn+r$ for some $0 \leq r 
< n$.  Alice can play the above good strategy in stages $1$ through $mn$ and 
plays arbitrarily in stage $mn+1$ through $nm+r$. Then Alice's gain in ${\Omega}_{nm+r}$ 
would be with high probability at least
\begin{equation}
 \frac{mn}{mn+r} \left( v - \epsilon \right) - 
\frac{rM}{mn + r},
\end{equation}
where $M$ is an upper bound on the gains. The above value is greater than 
$v - 2 \epsilon$ for $m$ large enough.

To sum this up, we have shown that there is a strategy for Alice (namely, i.i.d.\ $\tilde{\alpha}$)  that guarantees payoff $v$ for Alice, regardless of Bob's 
strategy. This implies that 
\begin{equation}
 \vstrong \geq v=\gvalue{{\Omega}_n}.
\end{equation}
which is the first part of our claim in equation \eqref{eqn22par1}.

Now using minimax expression for the Nash equilibrium we have
\begin{equation}
\begin{split}
 \gvalue{{\Omega}_n} &= \min_{\tilde{\beta}} 
\max_{\tilde{\alpha}} 
\frac{1}{n} \sum_{i=1}^n \ev{g_S(A_i,B_i)}\\
&= \min_{q(s)} \min_{\beta(b_i | s, a_{[i-1]}, b_{[i-1]})} 
\max_{\alpha(a_i |  a_{[i-1]}, b_{[i-1]})} \frac{1}{n} \sum_{i=1}^n
\ev{g_S(A_i,B_i)}\\
&= \min_{q(s)} \gvalue{\Gamma^{\emptyset, \onefunction}_n(q(s))}	,
\end{split}
\end{equation}
where in the second equality we have split Bob's (behavioral) strategy $\tilde{\beta}$ into two parts: first choosing the state, and then playing actions based on the chosen state and history of the game.
This completes the proof.
\end{proof}

\begin{proof}[Proof of Proposition \ref{prop:lower-bound-final}]
We have,
\begin{equation*}
  \begin{split}
    \vstrong(\mS, \emptyset, \onefunction) &\stackrel{(a)}{\geq} \min_{p(s)} \gvalue{\Gamma_n^{\emptyset, \onefunction}(p)} \\
    &\stackrel{(b)}{\geq} \min_{p(s)} \left ( \vex \osv{p} - \frac{C}{\sqrt{n}} \right ) \\
    & \stackrel{(c)}{\geq} \left ( \min_{p(s)}  \vex \osv{p} \right ) - \frac{C}{\sqrt{n}}\\
    & \stackrel{(d)}{=} \left ( \min_{p(s)}  \osv{p} \right ) - \frac{C}{\sqrt{n}},\\
  \end{split}
\end{equation*}
where $(a)$ uses Proposition~\ref{prop:lower-bound-sasl} (which holds for all values of $n$), $(b)$ uses Theorem~\ref{thm:zamir-side-information-on-one-side}, $(c)$ uses the fact that the constant $C$ is independent of $p$ and $(d)$ uses the fact that the minimum of the convex hull of the function is the same as the minimum of the function itself.

Since this holds for all values of $n$, the result is proved simply by sending $n$ to infinity.
\end{proof}

\begin{proof}[Proof of Proposition \ref{prop:uppernound}]
From equation \eqref{eqn:tbto1}, we have that 
\begin{equation}
  \label{eq:Bob-one-to-empty-upper-bound}
  \vweak(\mS, \emptyset, \onefunction) = \vweak(\mS, \emptyset, \emptyset).
\end{equation}
For any distribution $p(s)$, using Lemma~\ref{lem:support-invariant} and Remark~\ref{rem:vweak-vstrong-set-notation}, we have $\vweak(\mS, \emptyset, \emptyset) \leq \vweak(p(s), \emptyset, \emptyset)$, since $\support p \subset \mS$. Now we claim that $\vweak(p, \emptyset, \emptyset) \leq \osv{p}$. In order to do so, assume $v$ is a value that Alice can weakly guarantee when the state is generated from distribution $p$, $\mathscr T_A = \emptyset$ and $\mathscr T_B = \emptyset$. Therefore, due to the definition, for any $\epsilon>0$, with $n$ large enough, for any strategy $\beta_n$ for Bob in $\Gamma_n^{\emptyset, \emptyset}(p)$, there exists an strategy $\alpha_n$ for Alice such that $\pr{\sigma_n < v} < \epsilon$. Assume Bob plays the equilibrium strategy of $\osv{p}$, iid in $n$ games. Then since initially neither Alice nor Bob have any side information about the state, they do not gain any extra information by observing each other's strategies. Now, looking at the game at stage $k$, since Bob is playing his equilibrium strategy,  $\ev{g_S(A_k, B_k)} \leq \osv{p}$. Hence, 
\begin{equation*}
  \ev{\sigma_n} = \ev{\frac{1}{n} \sum_{i=1}^n g_S(A_i, B_i)} \leq \osv{p}.
\end{equation*}
On the other hand, $\pr{\sigma_n < v} < \epsilon$ implies $\ev{\sigma_n} \geq v(1-\epsilon)$. This together with the above inequality we have $v(1-\epsilon) \leq \osv{p}$. Since $\epsilon$ was arbitrary, $v \leq \osv{p}$ and thus $\vweak(\mS, \emptyset, \emptyset) \leq \osv{p}$. Since $p$ was arbitrary, by taking minimum over $p$ we get
\begin{equation*}
  \vweak(\mS, \emptyset, \emptyset) \leq \min_{p(s)} \osv{p}.
\end{equation*}
Substituting this into \eqref{eq:Bob-one-to-empty-upper-bound} finishes the proof.
\end{proof}

\section{An application}\label{compoundAVC}
In this section, we provide an application of the high probability framework. This section assumes a background in information theory. 
Consider an AVC channel with a legitimate sender/receiver and also an adversary. 
Assume that the channel has a state $S$ which is partially known to the 
encoder/decoder and the adversary (imperfect CSI). Communication channel is a conditional probability 
distribution $p(y|x,a,s)$ where $x$ is the encoder's input on the channel, $a$ 
is adversary's input on the channel, $s$ is the channel state and $y$ is 
the output at the decoder. We assume that $X, Y, A$ and 
$S$ take values in finite sets $\mX, \mY, \mA$ and $\mS$, respectively. We 
assume that the state $S$ is chosen from a distribution $p_S$. Encoder and 
decoder  both have the same side information $S_X$ about $S$, while the 
adversary has 
a side information $S_A$ about it. We assume that the channel state is chosen 
once and for 
all and remains unchanged during the consecutive uses of the channel (slow fading).  However, 
we assume that the channel noise in $p(y|x,a,s)$ is independent in different channel uses, \emph{i.e.,} $p(y_{[n]}|x_{[n]}, a_{[n]}, s)=\prod_{i=1}^np(y_i|x_i,a_i,s)$. Furthermore, as before without loss of generality we assume that $S_X$ and $S_A$ are functions of $S$, \emph{i.e.,}\ $S_X = \mathscr T_X(S)$ and 
$S_A = \mathscr T_A(S)$. We assume that $p(s_X)>0$ for all $s_X$.

Adversary observes the history of the game at any stage $i$, 
\emph{i.e.,} inputs put on the channel by the encoder $X_{[i-1]}$. Likewise, we assume 
that both the encoder and decoder observe adversary's input on the channel 
$A_{[i-1]}$. Therefore, this is a communication problem with feedback.

We assume that encoder and decoder have access to unlimited private shared
randomness, unknown to the adversary, allowing them to use randomized 
algorithms. A $(n, 2^{nR})$ code consists of strategies for encoding as 
well as strategies for decoding. The encoder wants to reliably send a message $M$ in 
$\{1, \dots, 2^{nR} \}$ 
via $n$ uses of the channel, while the adversary wants to prevent this from 
happening. 
More specifically, at stage $i$, the encoder creates input $X_i$ using the 
message $M$, its side information $S_X$, its shared randomness $K$, as well as 
$X_{[i-1]}, A_{[i-1]}$ previous transmissions by himself and the adversary. 
Therefore 
the encoder's strategy is to assign a probability to each symbol in $\mX$ given the history of 
the game.  Hence, 
$\alpha(x_i | x_{[i-1]}, a_{[i-1]}, s_X,m, k)$ which is 
the encoding strategy, determines the probability of encoder generating.
Adversary has also a strategy, which we denote by the conditional pmf $\beta(a_i | x_{[i-1]}, a_{[i-1]}, 
s_A, k_A)$ where $k_A$ denotes private randomness of adversary; it determines the probability of choosing $a_i$ as the input of the 
adversary, the history of the game and adversary's side information.  

At the decoder side, we find an $\hat{M}$ given $S_X, Y_{[n]}, A_{[n]}, K$; thus we are assuming that receiver observes $Y_{[n]}$ as well as adversary's inputs to the channel. The side information at the decoder is assumed to be $S_X$ which is the same as the one at the encoder. 
A rate $R$ is called achievable if for $\epsilon > 0$, there is some $N_0$ such that for any $n>N_0$, we can design encoding / decoding strategies such that independent of 
adversary's strategy, the probability of error, \emph{i.e.,}~$\pr{M \neq \hat{M}}$ is 
smaller than $\epsilon$. The supremum over all the achievable rates is called 
the capacity of the channel and is denoted by $C$. Our goal is to find $C$.
Figure~\ref{fig:avc-model} depicts our channel model.
Following the common assumption in the game theory literature, we 
assume that both encoder/decoder and adversary know each other's strategies. 
As in a repeated game with incomplete information, there is a tradeoff for both encoder and adversary to use or hide their side 
information about the channel state. 
\begin{figure}
 \centering
 \begin{tikzpicture}
 \tikzstyle{box}=[rounded corners,fill=blue!20,draw=blue!50,very thick]
 \tikzstyle{bigarrow}=[decoration={markings,mark=at position 0.999 with
{\arrow[scale=2]{stealth}}}, postaction={decorate}, shorten >=0.4pt]
 \draw[box] (-1,-0.5) rectangle (1,0.5);
 \node at (0,0) {\footnotesize$p(y|x,{\color{red} a},{\color{blue} s})$};
 \draw[bigarrow] (-2,0) -- (-1,0);
 \draw[bigarrow] (1,0) -- (4,0);
 
 \draw[dashed, rounded corners] (-5,-1.7) rectangle (2.5,3);
 
 \draw (3.1,1) -- (3.6,1) -- (3.6,0.3);
 \draw[bigarrow] (3.6,0.3) -- (4,0.3);
 \node[above] at (3.1,1) {$Y_1$};
 
 \draw (3.1,-1) -- (3.6,-1) -- (3.6,-0.3);
 \draw[bigarrow] (3.6,-0.3) -- (4,-0.3);
 \node[below] at (3.1,-1) {$Y_n$};
 
 \node at (3.1,0.65) {$\vdots$};
 \node at (3.1,-0.5) {$\vdots$};
 
 \node at (-1.5,0.3) {$X_i$};
 \node at (1.5,0.3) {$Y_i$};
 \draw[box] (-4,-0.5) rectangle (-2,0.5);
 \node at (-3,0) {Encoder};
 \draw[box] (4,-0.5) rectangle (6,0.5);
 \node at (5,0) {Decoder};
 \draw[bigarrow] (-6,0) -- (-4,0);
 \node at (-5.5,0.3) {$M$};
 \draw[bigarrow] (6,0) -- (7,0);
 \node at (6.5,0.3) {$\hat{M}$};
 
 \draw[bigarrow] (0,1.5) -- (0,0.5);
 \draw[box] (-1,1.5) rectangle (1,2.5);
 \node at (0,2) {{\color{red} Adversary}};
 \node at (0.3,1) {{\color{red} $A_i$}};
 
 \draw[bigarrow] (5,-1) -- (5,-0.5);
 \node at (5,-1.3) {{\color{blue} $S_X$}};
 
 \draw[bigarrow] (-3,-1) -- (-3,-0.5);
 \node at (-3,-1.3) {{\color{blue} $S_X$}};
 
 \draw[bigarrow] (-3,1) -- (-3,0.5);
 \node at (-3,1.3) {$A_{[i-1]}$};
 
 \draw[bigarrow] (5,1) -- (5,0.5);
 \node at (5,1.3) {$A_{[n]}$};

 \draw[bigarrow] (-1.5,2) -- (-1,2);
 \node at (-1.8,2.3) {$X_{[i-1]}$};
 
 \draw[bigarrow] (1.5,2) -- (1,2);
 \node at (1.8, 2.3) {{\color{blue} $S_A$}};
 
 
 
 \draw[green!70!blue,bigarrow] (5.5,-2.5) -- (5.5,-0.5);
 \draw[green!70!blue,bigarrow] (-3.5,-2.5) -- (-3.5,-0.5);
 \draw[green!70!blue] (-3.5,-2.5) -- (5.5,-2.5);
 \node[below, green!70!blue] at (1,-2.5) {Shared Randomness $K$};
\end{tikzpicture}
\caption{\label{fig:avc-model} A compound AVC channel.}
\end{figure}
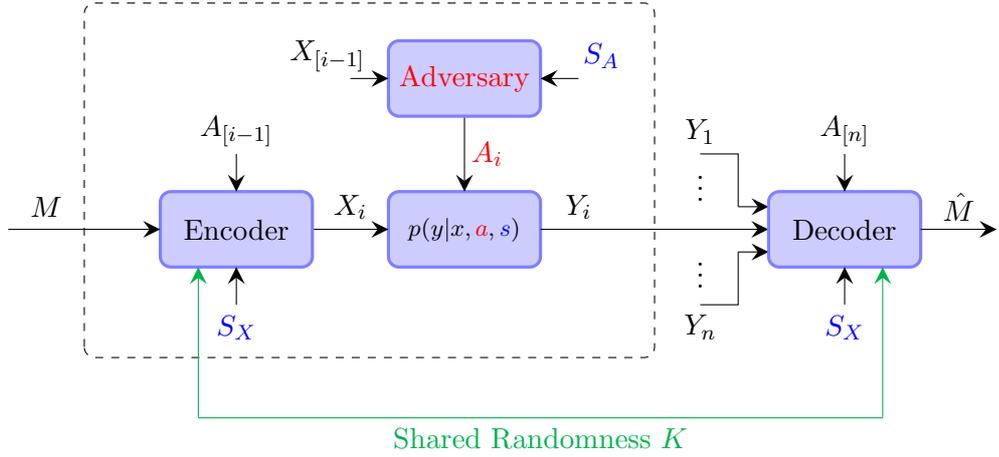

\begin{thm}
 For the Compound-AVC problem described above, the capacity is
 \begin{equation}
 \label{eq:capacity}
  C = \min_{s_X} \min_{p_S: \support(p_S) \subseteq \mathscr T_X^{-1}(s_X)}\min_{p(a)} \max_{p(x)} I(X;YA|S),
 \end{equation}
where $p(x,y,a,s)=p(s)p(a)p(x)p(y|a,s,x)$.
\end{thm}
The mutual information is between the input $X$, and $(Y, A)$ which are observed by the receiver. Observe that the expression does not depend on $S_A$.

\subsection{Converse}\label{converse}
For proving the converse, assume that the adversary puts its inputs i.i.d., from an 
input distribution $p(a)$ independent  of all its observations and its side information about the state. Then for a fixed 
value of state $s$, we have a point to point channel with input $X$ and output 
$YA$. 
The encoder receives the side information $S_X = s_X$; no further information about $S$ is revealed to him 
during the transmission, since adversary's input is independent of the state. 
Therefore, with the observation $S_X=s_X$ at the encoder and decoder, we have a classical memoryless compound channel with input $X$, output $YA$ and state $S$ with the conditional pmf $p(s|S_X=s)$. The capacity of this compound channel is \cite[Theorem 7.1]{ElGamalKim}
\begin{equation}
 \max_{p(x)} \min_{s: \mathscr T_X(s) = s_X} I(X;YA|S=s).
\end{equation}
Therefore
\begin{equation}
\begin{split}
 R &\leq \max_{p(x)} \min_{s: \mathscr T_X(s) = s_X} I(X;YA|S=s) \\
 &= \max_{p(x)} \min_{p_S: \support(p_S) \subseteq \mathscr T_X^{-1}(s_X)} I(X;YA{|S}) \\
 &\stackrel{(a)} = \min_{p_S: \support(p_S) \subseteq \mathscr T_X^{-1}(s_X)}\max_{p(x)} I(X;YA{|S}) \\
 \end{split}
\end{equation}
where $(a)$ results from the minimax theorem and the fact that 
$I(X;YA{|S})$ is 
concave in $p(x)$ and convex in $p(s)$. Since the above holds for all $s_X$ 
and also for all $p(s)$, 
\begin{equation}
 R \leq  \min_{p(a)}\min_{s_X} \min_{p_S: \support(p_S) \subseteq \mathscr T_X^{-1}(s_X)} \max_{p(x)} I(X;YA|S),
\end{equation}
since $p(S_X=s_X)>0$ for all $s_X$. This completes the proof of the converse.

\subsection{Achievability}\label{achieve}

\subsubsection{An auxiliary game}
Before specifying the encoder and decoder, we define an auxiliary repeated game with incomplete 
information  as follows: take $\mP$ to be a finite subset of the probability simplex 
$\Delta(\mX)$ over the input 
alphabet $\mX$. The game has two players: encoder/decoder (which we call 
encoder for the sake of simplicity) and adversary. The action set of encoder is $\mP$ and the action set of adversary is $\mA$. The one stage game 
has $| \mS |$ tables for each state of the channel. In payoff table 
corresponding to  $s \in \mS$, when encoder chooses action $\pi \in \mP$ and 
adversary chooses action $a \in 
\mA$, payoff $I(X;Y|S=s, A=a)$ for $p(x,y|s,a)=\pi(x)p(y|x,s,a)$
is assigned to the encoder (and its negative is assigned to the adversary). In the following, instead of writing $I(X;Y|s, a)$, we use $I(\pi;Y|s,a)$ in order 
to emphasize the dependence on $\pi$. Then this game 
is repeated $n$ times, and the total payoff function of encoder would be the sum of its individual payoffs from the $n$ games.  Further, we assume that the encoder and adversary receive $S_X$ and $S_A$ as 
their side information at the beginning of the game. We call this game $\Gamma_n$.

\subsubsection{From the auxiliary game to the compound-AVC problem}
Assume  that $\vstrong$ is the maximum value 
encoder can guarantee with high probability in the auxiliary game $\Gamma_n$. We claim that any rate $R < 
\vstrong$ is achievable for the original compound-AVC problem. Take some $\tilde R$ such that $R < \tilde R < \vstrong$. Assume the strategy of encoder for strongly guaranteeing $\tilde R$ is $p^E$. Thus, 
$p^E(\pi_i | s_X, a_{[i-1]}, \pi_{[i-1]})$ denotes the probability the 
encoder chooses distribution $\pi_i$ at stage $i$ given his observations up to 
that time. Adopting $p^E$,  the gain of the encoder in $\Gamma_n$ is at 
least $\tilde R$  with high probability when $n$ is large enough.

\textbf{Codebook generation:}
A codebook of $2^{nR}$ codewords of length $n$ can be illustrated by a table 
of size $2^{nR}\times n$ where row index indicates the message and columns 
indicate time steps. Encoder and decoder \emph{dynamically} construct the $2^{nR}\times n$ table, 
column by column, during the transmission process by running the auxiliary game in parallel. In other words, the column 
$i$ of the codebook (which is needed to make the $i$-th transmission) is created 
after time step $i-1$ as follows: the symbols in the $i$-th column of the codebook table are generated 
independently from 
distribution $\pi_i$ of the auxiliary game (\emph{i.e.,} $2^{nR}$ i.i.d. samples from 
$\pi_i$ 
are generated and put in the $i$-th column of the table). Note that since encoder 
and decoder have infinite shared 
randomness, they can use their shared randomness to simultaneously generate 
the codebook (\emph{i.e.,} the randomness needed to draw samples from $n$ i.i.d. samples 
from $\pi_i$ comes from the shared randomness between the encoder and decoder). 
The encoder and decoder are synchronized as the decoder observes 
$a_{[i-1]}$ and knows $S_X$.

\textbf{Encoding:} Having message $m$, the encoder sends the symbols from the $m$-th row of the codebook table that is being dynamically constructed during the transmission process. To write down the joint pmf that this encoding strategy implies, let $p^A$ denote
adversary's strategy in the compound-AVC problem, \emph{i.e.,} let $p^A(a_i | s_A, a_{[i-1]}, 
x_{[i-1]}, k_A)$ be the probability that adversary chooses $a_i$ at stage $i$ 
where $x_i$ is encoder's input on the channel at stage $i$ and $k_A$ is adversary's private randomness. Then, the 
joint distribution of variables in the problem when the state of the channel is 
$s$ and the message $m$ is
\begin{align*}
 p(s_X,s_A|s)p(k_A) &\prod_{i=1}^n p^E(\pi_i | s_X, \pi_{[i-1]}, a_{[i-1]}) 
p^A(a_i | s_A, a_{[i-1]}, x_{[i-1]}(m), k_A) 
\\&\times \left ( \prod_{j=1}^{2^{nR}} 
\pi_i(x_i(j)) \right ) p(y_i | x_i(m), s, a_i)
\end{align*}

\textbf{Decoding:}
The decoder has access to $a_{[n]}$, $y_{[n]}$. Also note that $\pi_i$ is generated 
from the strategy $p^E$, $S_X$, $\pi_{[i-1]}$ and $a_{[i-1]}$  which are 
all known to the decoder. Also as was mentioned above, since we use 
random strategies in the repeated game, $\pi_i$ is a random function of 
the observations. However, since encoder and decoder have access to 
shared randomness, they can use it to come up with the same $\pi_i$ and 
apply the strategy simultaneously.  Also since encoder and decoder have 
shared randomness, the decoder knows the codebook. For $\pi$ and $a$ in finite sets $\mathcal P$ and $\mathcal A$ respectively, define $\tau(\pi,a)$ to be the set 
of indexes $1 \leq i \leq n$ where encoder's distribution is $\pi$ and 
adversary's input is $a$, \emph{i.e.,}
\begin{equation}
 \tau(\pi,a) := \{ 1 \leq i \leq n: \Pi_i = \pi, A_i = a \}.
\end{equation}
Then in the decoder, assume the sequence $y_{[n]}$ is received. The receiver 
declares that message $\hat{m}$ has been sent if
\begin{equation}
\label{eq:decoding-condition}
 (x_{\tau(\pi,a)}(\hat{m}),y_{\tau(\pi,a)}) \in 
\mT_\epsilon^{n_{\pi,a}}(X,Y|a,s) \quad \forall 
\pi,a: n_{\pi,a} \geq n^{3/4} \quad  \text{for some } s \in 
\hat{\mS}(\pi_{[n]}, 
a_{[n]}),
\end{equation}
where $n_{\pi,a} = |\tau(\pi,a)|$ is the number of indexes $i$ where $\Pi_i = 
\pi$ and $A_i = a$; the set $\mT_\epsilon^{n_{\pi,a}}(X,Y|a,s)$ includes jointly typical sequences from $\mathcal{X}$ and $\mathcal{Y}$ of length $n_{\pi,a}$ according to $p(x,y|a,s)=\pi(x)p(y|x,a,s)$; and finally 
\begin{equation}
 \hat{\mS}(\pi_{[n]}, a_{[n]}) = \left \{ s \in \mS: \frac{1}{n} \sum_{i=1}^n 
I(\pi_i;Y|a_i,s) \geq \tilde R \right \},
\end{equation}

\textbf{Analysis of Error:}
Because the codebook is constructed symmetrically, without loss of 
generality we assume that $M=1$. We have two types of errors, the first one 
denoted by $\mE_1$ happens when $\hat{m} = 1$ does not satisfy 
\eqref{eq:decoding-condition} and $\mE_2$ happens when for some $\hat{m} \neq 
1$, \eqref{eq:decoding-condition} is satisfied.

For analyzing the first error, assume that $S = s^*$ has happened. First note 
that since encoder's strategy guarantees $\tilde R$, we have
\begin{equation}
\pr{ \frac{1}{n} \sum_{i=1}^n I(\Pi_i;Y|A_i,S) \geq \tilde R } \geq 
1-\epsilon,
\end{equation}
and hence
$
 \pr{S \in \hat{\mS}(\Pi_{[n]}, A_{[n]}) } \geq 1 - \epsilon.
$
So we can assume that
$
 s^* \in \hat{\mS}(\pi_{[n]}, a_{[n]}).
$
Hence, in order to show that $\hat{m} = 1$ satisfies 
\eqref{eq:decoding-condition}, we shall show that with high probability
\begin{equation}
\label{eq:error-1-typicallity}
(X_{\tau(\pi,a)}(1),Y_{\tau(\pi,a)})\in 
\mT_{\epsilon}^{n_{\pi,a}}(X,Y|a,s^*), \qquad \forall \pi, a: n_{\pi,a} \geq 
n^{3/4},
\end{equation}
In the above expression 
$s^*$ is the real state of the channel and $s_X^*$ and $s_A^*$ be the side informations. Iin the remaining we condition 
everything 
on $S=s^*$ and at times, we do not state this explicitly in our expressions for the sake 
of simplicity. 
Note that since adversary's input at stage $i$, $A_i$ is dependent on 
$X_{[i-1]}$, then we can not say that $X_{\tau(\pi,a)}$ are i.i.d.\ from 
distribution $\pi$. For instance, if $A_i = X_{i-1}$, then conditioned on our 
observations on adversary, the distribution on input is changed. Hence, we can 
not employ standard LLN type argument to show that the first error type 
vanishes. Instead, define
\begin{equation}
 W_i = N_i(a,\pi,x,y) - N_i(a,\pi) \pi(x) p(y|x,a,s^*) \qquad 1 \leq i \leq n,
\end{equation}
where
\begin{equation}
 N_i(a,\pi,x,y) = | \{ j \leq i: A_j = a, \Pi_j = \pi, X_j = x, Y_j = y \}|,
\end{equation}
is the number of times $a, \pi, x, y$ has happened up to stage $i$. 
Note that in the above definition, $a,\pi,x,y$ are fixed values and not 
random quantities.
Similarly
\begin{equation}
 N_i(a,\pi) = | \{ j \leq i: A_j = a, \Pi_j = \pi \} |.
\end{equation}
Also define $W_0 = 0$.
Now we claim that $W_i$ is a martingale with respect to $H_{[i]} := A_{[i]} \Pi_{[i]} X_{[i]} 
Y_{[i]}K_A$ which is the history of the events up to stage $i$. To see that note,
\begin{equation}
 \begin{split}
  \ev{W_{i+1} | H_{[i]}} &= W_i + \ev{\one{A_{i+1} = a, \Pi_{i+1} = \pi, X_{i+1} = 
x, Y_{i+1} = y} \big| H_{[i]}} \\
&\qquad - \ev{\one{A_{i+1} = a, \Pi_{i+1} = \pi} \pi(x) p(y|x,a,s^*) \big| H_{[i]}} \\
&= W_i + p^A(a|s_A^*, A_{[i]}, X_{[i]}, K_A) p^E(\pi |s^*_X, A_{[i]}, \Pi_{[i]}) \pi(x) p(y|x,a,s^*) \\
& \qquad - \pi(x) p(y|x,a,s^*) p^A(a|s^*_A, A_{[i]}, X_{[i]}, K_A) p^E(\pi |s^*_X, A_{[i]}, \Pi_{[i]}) \\
&= W_i,
 \end{split}
\end{equation}
where in the second equality we have used the fact that the expected value of 
an indicator function is the probability of its corresponding event. Hence, as 
was claimed, $W_i$ is a martingale. Also note that 
\begin{equation}
 |W_{i+1} - W_i| = \Big | \one{A_i=a, \Pi_i = \pi, X_i = x, Y_i = y} - 
\one{A_i = a, \Pi_i = \pi} \pi(x) p(y|a,x,s) \Big | \leq 1.
\end{equation}
Therefore using Azuma's inequality for $t = n^{3/4} \epsilon$, we have
\begin{equation}
 \pr{|W_n| \geq n^{3/4} \epsilon} \leq 2 \exp \left (-\frac{n^{3/2} 
\epsilon^2}{2n} \right ),
\end{equation}
which goes to zero as $n$ goes to infinity. Hence, for $n$ large enough with 
high probability we have
\begin{equation}
 \left | N_n(a,\pi,x,y) - N_n(a,\pi) \pi(x) p(y|x,a,s) \right | \leq 
n^{3/4} \epsilon,
\end{equation}
This statement is true for all $a,\pi, x, y$ which form a finite set. 
Therefore, we can take $n$ large enough so that the above expression is true 
wit high probability for all values of $a,\pi,x,y$. Now, if $N_n(\pi,a) \geq 
n^{3/4}$ we have
\begin{equation}
 \left | \frac{N_n(a,\pi,x,y)}{N_n(a,\pi)} - \pi(x) p(y|x,a,s) \right | \leq 
\frac{n^{3/4}\epsilon}{N_n(a,\pi)} \leq \epsilon,
\end{equation}
which shows that \eqref{eq:error-1-typicallity} is satisfied and the first type 
error vanishes as $n$ goes to infinity.

Now we analyze the second type of error. We condition the second error on 
$\Pi_{[n]} = \pi_{[n]}$ and $A_{[n]} = a_{[n]},Y_{[n]} = y_{[n]}$. Define $\mE_2(\hat{m})$ for 
$\hat{m}$ to be the event where $\hat{m}$ satisfies 
\eqref{eq:decoding-condition}. Note that since the adversary does not observe 
$X_{[n]}(\hat{m})$ for $\hat{m} \geq 1$, unlike the first type of error, it can not 
establish correlation between them. Thus, conditioned on $\pi_{[n]}$, $X_{[n]}(\hat{m})$ are independent and 
$X_i(\hat{m})$ is generated from $\pi_i$. Therefore for $s \in \hat{\mS}(\pi_{[n]}, 
a_{[n]})$, we can use packing lemma \cite[Lemma 3.1]{ElGamalKim}. Using the independence 
among blocks $\tau(\pi,a)$, for some $\pi,a$ where $n_{\pi,a} \geq n^{3/4}$ 
we have
\begin{equation}
 \pr{(X_{\tau(\pi,a)},Y_{\tau(\pi,a)}) \in \mT_\epsilon^{n_{\pi,a}} (X,Y|a,s) 
} \leq 2^{-n_{\pi,a} ( I(\pi,Y|a,s) - \delta(\epsilon))},
\end{equation}
for some $\delta(\epsilon)$ that converges to zero as $\epsilon$ converges to zero. Now using the independence of the above events, we have
\begin{align*}
 - \log \pr{(X_{\tau(\pi,a)},Y_{\tau(\pi,a)}) \in \mT_\epsilon^{n_{\pi,a}} 
(X,Y|a,s) ~~~\forall \pi,a: n_{\pi,a} \geq n^{3/4} } \\\geq \sum_{\pi, a: 
n_{\pi,a} \geq n^{3/4}} n_{\pi,a} (I(\pi;Y|a,s) - \delta(\epsilon)).
\end{align*}
Now since the mutual information is bounded and the terms corresponding to 
those $\pi_i, a_i$ that do not appear in the above expression have length 
less than $n^{3/4}$, and the set of possible $\pi,a$ is finite, there is a 
bounded constant $\bar M$ such that
\begin{equation}
\begin{split}
 \sum_{\pi, a: 
n_{\pi,a} \geq n^{3/4}} n_{\pi,a} (I(\pi;Y|a,s) - \delta(\epsilon)) &\geq 
\sum_i I(\pi_i;Y|a_i,s) - n \delta(\epsilon) - \bar M n^{3/4} \\
&\geq n(\tilde R - \delta(\epsilon) - \bar M n^{-1/4} ),
\end{split}
\end{equation}
where the last inequality uses the assumption $s \in \hat{\mS}(\pi_{[n]}, a_{[n]})$. 
Therefore using union bound
\begin{equation}
\label{eq:whysetv}
 \pr{\mE_2} \leq 2^{nR} 2^{-n(\tilde R - \delta(\epsilon) - \bar M n^{-1/4} )} = 2^{-n(\tilde R - 
R - \delta(\epsilon) - \bar M n^{-1/4})},
\end{equation}
the above value goes to zero as $n$ goes to infinity by appropriate choice of 
$\epsilon$ since $\tilde R > R$.

Hence we have proved that any rate below $\vstrong$ is achievable. 

\subsubsection{Computing $\vstrong$ for the auxiliary game}
In the rest of the proof, we use Theorem \ref{thm:main-two-sided} to find the value of 
$\vstrong$. We need to first find 
$\osv{p}$, which is the game value for the average payoff table 
$
 \sum_s p(s) I(\pi;Y|a,s)
$. 
Thus,
\begin{equation}
 \osv{p} = \min_{p(a)} \max_{\pi \in \mP} \sum_s p(s)p(a) I(\pi;Y|a,s).
\end{equation}
Writing $X$ instead of its distribution $\pi$ for convenience we have
\begin{equation}
 \osv{p(s)} = \min_{p(a)} \max_{p(x)\in \mP} I(X;Y|AS),
\end{equation}
where the joint distribution of the variables is $p(s) p(x)p(a) p(y|x,a,s)$. 
Since $X,A$ and $S$ are independent, we have
\begin{equation}
 \osv{p(s)} = \min_{p(a)} \max_{p(x)\in \mP} I(X;YA|S).
\end{equation}
Using Theorem \ref{thm:main-two-sided}, we have
\begin{equation}
 \vstrong = \min_{s_A} \min_{p_S: \support(p_S) \subseteq \mathscr T_A^{-1}(s_A)} \min_{p(a)} \max_{p(x)\in \mP} 
I(X;YA|S).
\end{equation}

In the above argument, the set $\mP$ is a finite and arbitrary subset of 
distributions on $\mX$. Now the only thing which remains to show is that by 
appropriate choice of finite set $\mP$ we can get arbitrarily close to the 
target value in  \eqref{eq:capacity}. In order to do so, define function $f$ as 
\begin{equation}
 f(p(x),p(a),p(s)) = I(X;YA|S),
\end{equation}
where the joint distribution is $p(s)p(x)p(a)p(y|xas)$. This function is 
continuous on the product of compact spaces which is compact itself. Therefore, 
$f$ is uniformly continuous. Hence, since the set of distributions on $\mX$ is 
compact, for every given, $\epsilon > 0$, there is a finite covering $\mP$ of 
$\Delta(\mX)$ where for all $p(x) \in \Delta(\mX)$, there exists $\tilde{p}(x) \in \mP$ such that for all $
p(s),p(a)$
\begin{equation}
| f(p(s), p(x), p(a)) - f(p(s), \tilde{p}(x),p(a)) | \leq 
\epsilon.
\end{equation}
Therefore by 
appropriate choice of finite set $\mP$ we can get within any $\epsilon$ to the 
target value in  \eqref{eq:capacity}.

\end{document}